\documentclass[journal,twoside,web]{_Aux/ieeecolor}
\usepackage{times}
\usepackage{setspace}
\usepackage{url}
\spacing{1}
\usepackage[utf8]{inputenc}
\usepackage[T1]{fontenc}
\usepackage{graphicx}		
\usepackage{wrapfig}
\usepackage[format=plain,font=footnotesize,labelfont=bf,labelsep=period]{caption}
\usepackage{sidecap} 
\usepackage[export]{adjustbox}
\usepackage{subcaption}
\usepackage[font=small]{caption}
\usepackage{float}

\usepackage{amsmath} 
\usepackage{amssymb}  
\usepackage{amsthm}
\usepackage{mathtools}
\usepackage[normalem]{ulem}
\usepackage{paralist}	
\usepackage[space]{grffile} 
\usepackage{color}
\let\labelindent\relax
\usepackage{enumitem}
\usepackage{bm}
\usepackage{cancel}
\usepackage{hhline}
\usepackage[c2 , nocomma]{optidef}
\usepackage{oubraces}
\usepackage{algorithmic}
\usepackage{graphicx}
\usepackage{textcomp}

\usepackage{amsfonts}
\usepackage{lipsum}

\usepackage{lcsys}
\usepackage{cite}
\usepackage{hyperref}

\newtheorem{theorem}{Theorem}
\newtheorem{corollary}{Corollary}

\theoremstyle{definition}
\newtheorem{definition}{Definition}
\theoremstyle{remark}
\newtheorem{remark}{Remark}
\theoremstyle{definition}

\theoremstyle{definition}

\newcommand{\C}{\mathcal{C}}

\definecolor{darkblue}{RGB}{0,0,102}
\definecolor{lightblue}{RGB}{77,77,148}

\definecolor{gold}{RGB}{234, 170, 0}
\definecolor{metallic_gold}{RGB}{139, 111, 78}

\DeclareMathOperator*{\esssup}{ess\,sup}

\DeclareMathOperator{\diag}{diag}

\DeclareMathOperator*{\argmin}{arg\,min}

\usepackage{xfrac}

\usepackage{dblfloatfix}

\usepackage{pdfpages}
\usepackage{graphicx}

\usepackage{enumitem}

\def\BibTeX{{\rm B\kern-.05em{\sc i\kern-.025em b}\kern-.08em
    T\kern-.1667em\lower.7ex\hbox{E}\kern-.125emX}}

\begin{document}

\title{
Rollover Prevention for Mobile Robots with Control Barrier Functions: Differentiator-Based Adaptation and Projection-to-State Safety
}

\author{Ersin Da\c{s}, \IEEEmembership{Member, IEEE}, Aaron D. Ames, \IEEEmembership{Fellow, IEEE}, and Joel W. Burdick, \IEEEmembership{Member, IEEE}
\thanks{*This work was supported by DARPA under the LINC program.}
\thanks{The authors are with the Department of Mechanical and Civil Engineering, California Institute of Technology, Pasadena, CA 91125, USA. ${\tt\small \{ersindas, ames, jburdick \}@caltech.edu}$ } }

\maketitle
\thispagestyle{empty}         

\begin{abstract}
This paper develops rollover prevention guarantees for mobile robots using control barrier function (CBF) theory, and demonstrates the method experimentally. We consider a safety measure based on a zero moment point condition through the lens of CBFs. However, these conditions depend on time-varying and noisy parameters. To address this issue, we present a differentiator-based safety-critical controller that estimates these parameters and pairs Input-to-State Stable (ISS) differentiator dynamics with CBFs to achieve rigorous safety guarantees. Additionally, to ensure safety in the presence of disturbances, we utilize a time-varying extension of Projection-to-State Safety (PSSf). The effectiveness of the proposed method is demonstrated via experiments on a tracked robot with a rollover potential on steep slopes.
\end{abstract}


\section{Introduction} 
\label{sec:intro}
\IEEEPARstart{A}{utonomous} robotic systems are increasingly deployed in complex and real-world environments, prompting a corresponding rise in the importance of developing safety-critical control methods \cite{medeiros2020}. As mobile robots often operate on uneven terrains and in dynamic conditions, preventing rollover is a vital aspect of their design and operation \cite{jeon2023planned}. Improved rollover safety not only improves the overall safety profile of mobile robots but also significantly contributes to their reliability and effectiveness in real-world applications.

Several methods measure the risk of rollover in mobile robots, including stability measures like force-angle stability, moment-height stability, and zero moment point (ZMP) \cite{roan2010real}. Leveraging these characterizations, a variety of control techniques have been developed to prevent rollovers: nonlinear programming \cite{de2019trajectory}, chance-constrained optimal control \cite{song2023chance}, and invariance control \cite{lee2012rollover}. These methods often rely on high-fidelity models or require numerous sensors, which may limit their practical applicability in real-world scenarios. The goal of this paper is to develop a new approach for rollover avoidance that is both rigorous, but also implementable. 

Safety is often framed as forward set invariance; guaranteeing that system states stay within a predetermined set ensures system safety. Control barrier functions (CBFs) \cite{ames2017control} have emerged as a tool for synthesizing controllers that guarantee forward invariance of a given safe set. The CBF framework also leads to safety filters, which have been successfully applied in various domains \cite{wabersich2023data}. These filters alter control inputs only when necessary for safety. However, accurate system dynamic models are needed for safety guarantees when controllers are synthesized via CBFs. Thus, the presence of unmodeled system dynamics causes uncertainty in the CBF condition, potentially leading to safety constraint violation. 
\begin{figure}[t]
    \centering
\includegraphics[width=1\linewidth]{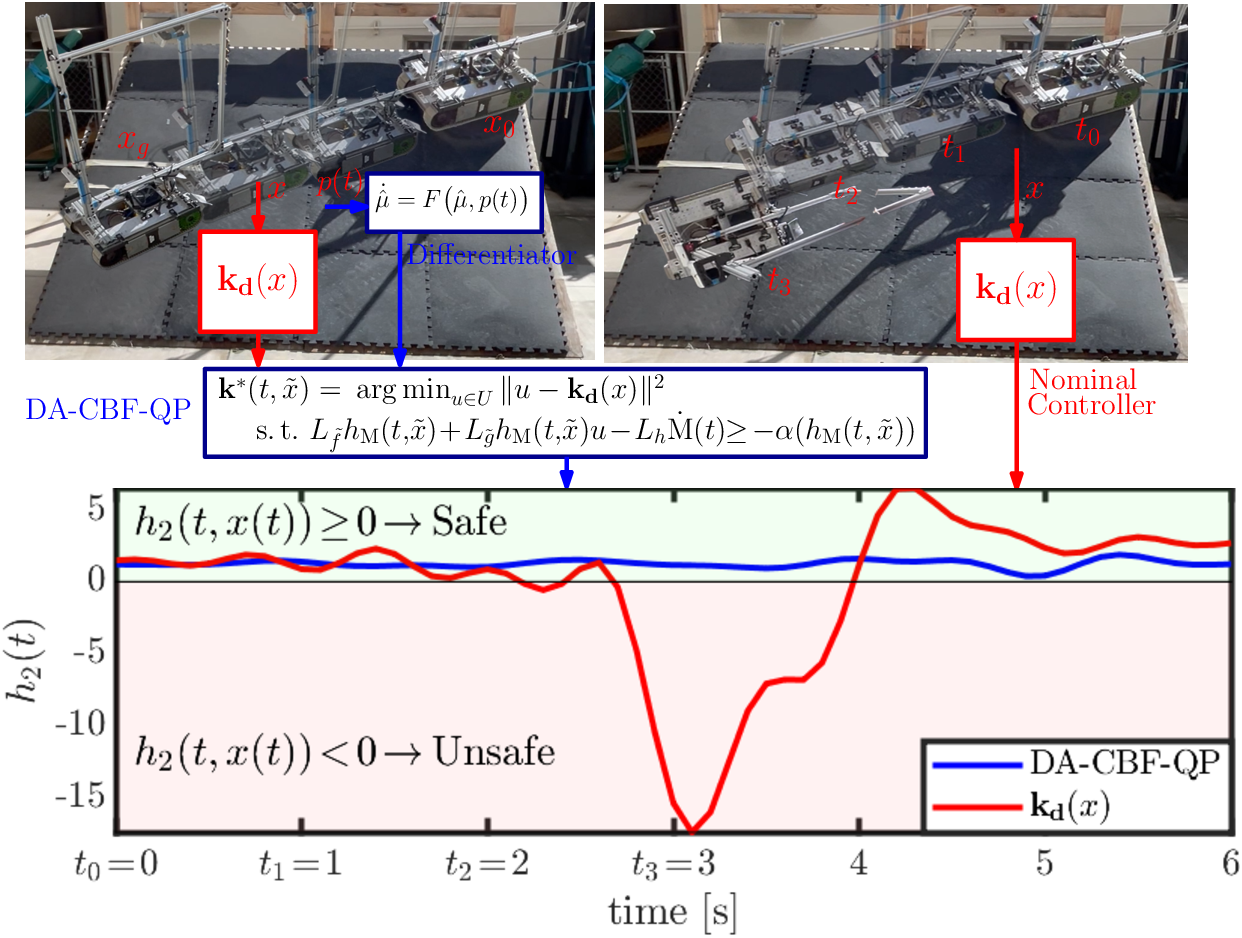}
    \caption{Experimental results for robot rollover prevention. The proposed DA-CBF-QP safety filter maintains safety (Left, video montage of robot motion). However, under the nominal controller, the robot leaves the safe set (Right). The value of CBF $h$ vs. time (Bottom).
    }
   \label{fig:main}
	\vskip - 6mm
\end{figure}

Projection-to-State Safety (PSSf) \cite{taylor2020} builds upon the notion of Input-to-State Safety (ISSf) \cite{kolathaya2018} to establish a framework for quantifying the effect of uncertainty or disturbances on safety guarantees. But ZMP-based rollover constraints require estimates on the noisy derivative of the gravity vector. Safety-critical control in dynamic environments via CBFs was proposed in \cite{molnar2021safety} by using constant worst-case bounds for the time-varying parameters, which may result in undesired conservativeness. CBFs coupled with estimators can address the moving obstacles avoidance problem \cite{tezuka2020}. However, the extension to address broader dynamic parameter-dependent safe control design problems has not yet been considered.

This paper presents a framework for synthesizing safety filters that are robust to time-varying parameters. We introduce \textit{differentiator-adaptive CBFs} (DA-CBFs) that consider the time-derivatives of time-varying parameters that are necessary to enforce CBF conditions. When the differentiator dynamics are ISS with respect to noise, the result is a new time-varying CBF whose satisfaction ensures safety. Moreover, to address model uncertainty in the time derivative of a time-varying CBF, we define an extension of PSSf, \textit{time-varying PSSf} (tPSSf). The main result gives conditions on DA-CBFs such that PSSf is guaranteed. Practically, these contributions enable robust rollover prevention for mobile robots via the synthesis of CBFs from ZMP constraints. We validate the efficacy of the proposed approach through experiments on a tracked mobile robot (Fig~\ref{fig:main}) encountering rollover issues triggered by slopes.

\section{Preliminaries} 
\label{sec:pre}
Consider a nonlinear control affine system of the form:
\begin{equation}
\label{system}
    \dot{x}  = f(x) + g(x) u,
\end{equation}
where ${x \!\in\! X \!\subset\! \mathbb{R}^n}$ is the state, ${f\!:\! X \!\to\! \mathbb{R}^n}$, ${g\!:\! X \!\to\! \mathbb{R}^{n \times m} }$ are locally Lipschitz continuous on the open and connected set $X$, and ${u \!\in\! U \!\subset\! \mathbb{R}^m}$ is the control input. A locally Lipschitz continuous controller ${u \!=\! \mathbf{k}(x)}$, with ${\mathbf{k}\!:\! X \!\to\! U}$, yields a locally Lipschitz continuous \textit{closed-loop} control system, ${f_{\rm cl}\!:\! X \!\to\! \mathbb{R}^n}$:
\begin{equation}
\label{eq:clsystem1}
    \dot{x} = {f}(x) + {g}(x) \mathbf{k}(x) \triangleq f_{\rm cl}(x).
\end{equation}
Hence, given any initial condition ${x_0 \!\triangleq\! x(t_0) \!\in\! X}$ there exists an interval $\mathcal{I} ({x}_{0} ) \!\triangleq\! \left[t_0, t_{\max }\right)$ such that 
\begin{equation}
\label{eq:flow1}
     {x}(t) = x_0 + \int^{t}_{t_0} f_{\rm cl}(x(\tau)) d\tau ,~t > t_0
\end{equation}
is the unique solution to \eqref{eq:clsystem1} for ${t \!\in\! \mathcal{I} ({x}_{0} )}$; see \cite{daleckii2002}. Throughout this study we assume ${f_{\rm cl}}$ is forward complete, i.e., ${\mathcal{I} ({x}_{0} ) \!=\! [0, \infty)}$, and ${U}$ is a convex polytope. 

In this paper, the system is considered safe as long as its defined state remains within a non-empty set ${\C \!\subset\! X }$. In particular, let the set ${\C }$ be the 0-superlevel set of a continuously differentiable function ${h\!:\! X \!\to\! \mathbb{R}}$:
\begin{align*}
\begin{split}
\label{CBF1}
    \C \!\triangleq\! \left\{ x \!\in\! X \!\subset\! \mathbb{R}^n \!:\! h(x) \!\geq\! 0 \right\}, \\
    \partial \C \!\triangleq\! \left\{ {x \!\in\! X \!\subset\! \mathbb{R}^n} \!:\! h(x) \!=\! 0 \right\}. 
\end{split}
\end{align*}
This set is \textit{forward invariant} if, for any initial condition ${x(0) \!\in\! \C}$, the solution \eqref{eq:flow1} satisfies ${x(t) \!\in\! \C, ~\forall t \!\geq\! 0}$. The closed-loop system \eqref{eq:clsystem1} is {\em safe} on the \textit{safe set} $\C$ if $\C$ is forward invariant. CBFs \cite{ames2017control} have been proposed to synthesize safety-critical controllers that can ensure forward invariance. 
\begin{definition}[Control Barrier Function, \cite{ames2017control}]
\label{def:cbf}
Let ${\C \!\subset\! X }$ be the 0-superlevel set of a continuously differentiable function ${h\!:\! X \!\to\! \mathbb{R}}$. 
The function $h$ is a \textit{control barrier function} for system (\ref{system}) on $\C$ if ${\frac{\partial h}{\partial x} \neq 0}$ for all ${ x \!\in\! \partial \C}$ and there exists an extended class-$\mathcal{K}_{\infty}$ function\footnote{ 
A continuous function ${\alpha \!:\! [0, a ) \!\to\! \mathbb{R}^+}$, where ${a \!>\! 0}$, belongs to class-${\mathcal{K}}$ (${\alpha \!\in\! \mathcal{K}}$) if it is strictly monotonically increasing and ${\alpha(0) \!=\! 0}$. And, ${\alpha}$ belongs to class-${\mathcal{K}_\infty}$ (${\alpha \!\in\! \mathcal{K}_\infty}$) if ${a \!=\! \infty}$ and ${\lim_{r \!\to\! \infty} \alpha(r) \!=\! \infty.}$ A continuous function ${\alpha \!:\! \mathbb{R} \!\to\! \mathbb{R}}$ belongs to the set of extended class-$\mathcal{K}_\infty$ functions (${\alpha \!\in\! \mathcal{K}_{\infty, e}}$) if it is strictly monotonically increasing, ${\alpha(0) \!=\! 0}$, ${\lim_{r \!\to\! \infty} \alpha(r) \!=\! \infty}$ and ${\lim_{r \!\to\! -\infty} \alpha(r) \!=\! -\infty}$. A continuous function ${\beta \!:\! [0, a ) \!\times\! \mathbb{R}^+_0 \!\to\! \mathbb{R}^+_0}$ belongs to class-${\mathcal{K L}}$ (${\beta \!\in\! \mathcal{KL}}$), if for every ${s \!\in\! \mathbb{R}^+_0}$, ${\beta(\cdot, s)}$ is a class-$\mathcal{K}$ function and for every ${r \!\in\! [0, a )}$, ${\beta(r, \cdot) }$ is decreasing and ${\lim_{s \!\to\! \infty} \beta(r, s) \!=\! 0.}$} $\alpha \!\in\! \mathcal{K}_{\infty, e}$ such that for all ${x \!\in\! \C}$:
\begin{equation}
\label{cbf}
   \sup_{u \in U} \big [ {\dot{h}(x, u)}  \big ] 
   \!=\! \sup_{u \in U} \big [ L_f h(x) + L_g h(x) u  \big ] 
   \!\geq\! -\alpha (h(x)),
\end{equation}
where ${L_f h \!:\! X \!\to\! \mathbb{R}}$, ${L_g h \!:\! X \!\to\! \mathbb{R}^m}$ are Lie derivatives. 
\end{definition}
Note that, when ${U \!=\! \mathbb{R}^m}$, \eqref{cbf} is equivalent to:
\begin{equation}
\label{eq:bbfcheck}
\forall x\! \in\! X\!: \ L_{g} h(x)=0 \implies L_{f} h(x) \geq-\alpha(h(x)),
\end{equation}
which implies that if ${U \!=\! \mathbb{R}^m}$, satisfaction of condition \eqref{cbf} at states where ${L_{g} h(x)\!=\!0}$ is necessary and sufficient for the verification of a CBF. We note that for a bounded control input, i.e., ${u \!\in\! U \!\subset\! \mathbb{R}^m}$, \eqref{eq:bbfcheck} is a necessary (but not sufficient) condition for \eqref{cbf}.

Given a CBF ${h}$ and a corresponding ${\alpha}$ for \eqref{system}, the pointwise set of all control values that satisfy \eqref{cbf} is given by
\begin{equation*}
    K_\text{CBF}(x) \triangleq \big\{ u \in U  \big |  \dot{h}(x, u) \geq - \alpha (h(x)) \big\}.
\end{equation*}
We can establish formal safety guarantees based on Definition~\ref{def:cbf} with the help of the following theorem \cite{ames2017control}:
\begin{theorem}
    \label{teo:cbfdef}
    If $h$ is a CBF for \eqref{system} on $\C$ with an ${\alpha \!\in\! \mathcal{K}_{\infty, e}}$, then any Lipschitz continuous controller ${\mathbf{k}\!:\! X \!\to\! U}$ satisfying 
    \begin{equation}
        \label{eq:cbf_def}
        \dot{h}\left(x, \mathbf{k}(x)\right) \geq - \alpha (h(x)),~~\forall x \in \C ,
    \end{equation}
    renders \eqref{eq:clsystem1} safe with respect to $\C$. 
\end{theorem}
Given a baseline (possibly unsafe) locally Lipschitz continuous \textit{nominal controller} ${\mathbf{k_d} \!:\! X \!\to\! U}$, and a CBF ${h}$ with a corresponding ${\alpha}$ for system \eqref{system}, safety can be ensured by solving the CBF-Quadratic Program (CBF-QP) \cite{ames2017control}:   
\begin{align*}
\begin{array}{l}
{\mathbf{k^*}(x)= \ }
\displaystyle  \argmin_{{u}  \in U} \ \ \ {\|u-\mathbf{k_d}(x) \|^2}  \\ [1mm]
~~~~~~~~~~~~\textrm{s.t.} ~~~~~~~~~\dot{h}(x, u)  \geq - \alpha (h(x)),
\end{array}
\end{align*}
which enforces ${\mathbf{k^*} \!:\! X \!\to\! U}$ to take values in ${K_\text{CBF}(x)}$; thus, CBF-QP is also called a {\em safety filter}. 
If ${\mathbf{k^*}(x) \!\in\! K_\text{CBF}(x) }$ for all ${x \!\in\! X}$, then the set ${\C}$ is asymptotically stable for the forward complete closed-loop system ${f_{\rm cl}}$ in ${X}$ \cite{ames2017control}. 

\section{Main Result} 
\label{sec:theory}
This section first defines tPSSf to consider model uncertainty in the time derivative of a time-varying CBF. Then, we introduce DA-CBFs. 

In practice, control systems face uncertainties and disturbances that cannot be fully modeled. Thus, we consider a disturbed nonlinear control affine system:
\begin{equation}
\label{eq:systemdist}
    \dot{x}  = f(x) + g(x) u + d(t),
\end{equation}
where ${d\!:\! \mathbb{R}^+_0 \!\to\! \mathbb{R}^n}$ is the disturbance that can alter the safety property endowed by the CBF for system \eqref{system}. 

For the sake of generality, we consider a time-varying continuously differentiable function ${h \!:\! \mathbb{R}^+_0 \!\times\! X \!\to\! \mathbb{R} }$, and its 0-superlevel set given by
\begin{equation}
    \label{eq:tCBF}
    \C(t) \triangleq \big\{  x \!\in\!  X : h(t, x)  \geq 0 \big\},
\end{equation}
with ${\partial \C(t) \!\triangleq\! \{ {x \!\in\!  X} \!:\! h(t,x) \!=\! 0 \}}$.

\subsection{Time-Varying Projection-to-State Safety}
\label{subsec:tPSSf}
We assume that the effect of the disturbance $d$ on the derivative of CBF $h$, termed a \textit{projected disturbance}, is bounded:  
\begin{equation}
    \label{eq:distBound}
    \delta(t,x) \!\triangleq\! \dfrac{\partial h(t,x)}{\partial x} d(t); \quad |\delta(t, x(t))| \!\leq\! \bar{\delta}(t) ,
\end{equation}
where $\bar{\delta} \!:\! \mathbb{R}^+_0 \!\to\! \mathbb{R}^+_0$. Using this upper bound, we consider a time-varying set $\C_{\delta}(t)$ such that for all ${t \!\geq\! 0}$, $\C_{\delta}(t) \!\subset\! \C(t)$: 
\begin{equation}
    \label{eq:subset}
    \C_{\delta}(t) \triangleq \big\{  x \!\in\!  X :  \Bar{h}(t,x) \triangleq h(t, x) - \bar{\delta}(t)  \geq 0 \big\}.
\end{equation}
This leads to the following: 

\begin{definition}[Time-Varying Projection-to-State Safety]
\label{def:tpssf}
Given a state feedback controller ${\mathbf{k}\!:\! X \!\to\! U}$, the closed-loop system with the disturbance input ${\dot{x} \!=\! {f}(x) \!+\! {g}(x) \mathbf{k}(x) \!+\! d(t) }$ is \textit{time-varying projection-to-state safe} (tPSSf) on $\C_{\delta}(t)$ with respect to the function ${\Bar{h}\!:\! \mathbb{R}^+_0 \!\times\! X \!\to\! \mathbb{R} }$ and bounded projected disturbance ${\delta}$ if there exists ${\bar{\delta}(t)}$ such that ${\C(t) \!\supset\! \C_{\delta}(t)}$ is forward invariant for all ${t \!\geq\! 0}$.
\end{definition}

\begin{remark}
\label{re:tPSSF}
PSSf, proposed in \cite{taylor2020}, characterizes safety in the presence of a disturbance or model uncertainty using a time-invariant bound ${|\delta|_\infty \!\triangleq\! \esssup_{t \geq 0}{|\delta(t, x(t)| \!\leq\! \bar{\delta}}}$ in \eqref{eq:subset}. Moreover, PSSf defines a larger forward invariant set, given by ${\C_{\delta_\infty} \!\triangleq\! \{ x  \!\in\!  X \!:\! h(x) \!+\! |\delta|_\infty  \!\geq\! 0 \}}$, ${\C \!\subset\! \C_{\delta_\infty}}$, with a time-invariant function $h$. Thus, the system can leave the safe set ${\C}$ while remaining within the larger set ${\C_{\delta_\infty}}$. On the other hand, Definition~\ref{def:tpssf} utilizes the time-varying bound ${\bar{\delta}(t)}$ to consider the projected disturbance, and defines a smaller time-dependent forward invariant set $\C_{\delta}(t)$ to guarantee that the system stays in the original set $\C(t)$. Note that disturbance observer-based robust CBF methods \cite{dacs2022robust} utilize a time-varying bound, which is provided by the disturbance observer, with a corresponding subset definition similar to \eqref{eq:distBound} and \eqref{eq:subset}.  
\end{remark}

Next, given the set $\C_{\delta}(t)$, using Definition~\ref{def:tpssf}, the following theorem ensures the forward invariance of the original set $\C(t)$ in the presence of a disturbance:
\begin{theorem}
\label{theo:main}
Let ${\C_{\delta}(t) }$ given in \eqref{eq:subset} be the $0$-superlevel set of a continuously differentiable function ${\Bar{h}\!:\! \mathbb{R}^+_0 \!\times\! X \!\to\! \mathbb{R}}$ with $0$ as a regular value. Any locally Lipschitz continuous controller ${\mathbf{k}\!:\! X \!\to\! U}$ satisfying
\begin{equation}
\label{eq:mainth}
L_f \Bar{h} (t,x) + L_g \Bar{h} (t,x) \mathbf{{k}}(x) + \dfrac{\partial \Bar{h}(t,x)}{\partial t} \!\!\geq\! -\alpha (\Bar{h}(t,x)\! ),
\end{equation}
for all ${x(t) \!\in\! \C_{\delta}(t) }$ renders the disturbed system \eqref{eq:systemdist} tPSSf on ${\C_{\delta}(t)}$ with respect to the projected disturbance ${{\delta} \!:\! \mathbb{R}^+_0 \!\times\! X \!\to\! \mathbb{R}}$ if there exists a time-varying function ${\bar{\delta} \!:\! \mathbb{R}^+_0 \!\to\! \mathbb{R}^+_0}$ satisfying ${|\delta(t, x(t))| \!\leq\! \bar{\delta}(t)}$ and ${\alpha \!\in\! \mathcal{K}_{\infty, e}}$ such that for all ${t \!\geq\! 0}$:
\begin{equation}
\label{eq:classp}
-\Dot{\bar{\delta}}(t) + \bar{\delta}(t) \!\leq\! -\alpha (-\bar{\delta}(t)).
\end{equation}
\end{theorem}
\begin{proof}
\label{pro:main1}
Our goal is to show that the set ${\C(t)}$ is forward invariant. From \eqref{eq:subset}, \eqref{eq:mainth}, \eqref{eq:classp} and the time derivative of ${h\!=\!\Bar{h}\!+\!\Bar{\delta}}$ along the disturbed system \eqref{eq:systemdist} we have:
\begin{align}
\begin{split}
\label{eq:mainpr}
\!\!\Dot{{h}} &\!=\! L_f \Bar{h} (t,\!x) \!+\! L_g \Bar{h} (t,\!x) \mathbf{{k}}(x) \!+\! \dfrac{\partial \Bar{h}(t,\!x)}{\partial t} \!+\! \Dot{\bar{\delta}}(t) \!+\! \delta(t,\!x) \\
&\geq L_f \Bar{h} (t,\!x) \!+\! L_g \Bar{h} (t,x) \mathbf{{k}}(x) \!+\! \dfrac{\partial \Bar{h}(t,x)}{\partial t} \!+\! \Dot{\bar{\delta}}(t)  \!-\! \bar{\delta}(t) \\
&\geq -\alpha (\Bar{h}(t,x)) \!+\! \Dot{\bar{\delta}}(t) - \bar{\delta}(t) \\
&\geq -\alpha (\Bar{h}(t,x)) + \alpha (-\bar{\delta}(t)) \\
&= - \left ( \alpha ({h}(t,x) - \bar{\delta}(t)) + \alpha (-\bar{\delta}(t)) \right ). 
\end{split}
\end{align}
Next, we consider a choice state such that ${ x(t) \!\in\! \partial \C(t)}$, i.e., ${{h}(t,x)\!=\! 0}$, for which \eqref{eq:mainpr} implies ${\Dot{{h}} \geq 0}$. And we have ${\frac{\partial h(t,x)}{\partial x} \neq 0}$ for all ${x(t) \!\in\! \partial {\C(t)}}$ from $0$ as a regular value assumption. 
Therefore, Nagumo's theorem \cite{blanchini2008set} guarantees that ${ h(0, x(0)) \geq 0 \implies h(t, x(t)) \geq 0, \forall t \geq 0}$.
\end{proof}

\vskip -0.1 true in
\subsection{Safety with Differentiator-based CBF{s}} 
\label{subsec:theory}
When noisy parameter measurements impact safety constraints, a differentiator can estimate necessary time derivatives for CBF conditions, such as the ZMP constraints, that depend on noisy acceleration (gravity) measurements.

Let ${p_0(t)}$ with ${p_0\!:\! \mathbb{R}^+_0 \!\to\! \mathbb{R}}$ be a continuously differentiable function with a globally Lipschitz continuous time derivative. A measurable noisy signal ${p\!:\! \mathbb{R}^+_0 \!\to\! \mathbb{R}}$ can be written as ${p(t) \!=\! p_0(t) \!+\! v(t)}$, where $v$ is a bounded signal: ${\|v(t)\|_{\infty} \!\triangleq\! \sup_{t} {\|v(t)\|} \!<\! \infty}$, denoted by ${v \!\in\! L_{\infty}}$.

The main goal of a differentiator is to estimate ${\dot{p}_0(t)}$ for all ${t \!\geq\! 0}$ by taking ${p(t)}$ as an input. The dynamics ${\dot{p}_0(t)}$ are a single-input single-output system in \textit{strict feedback form}:
\begin{align}
    \label{eq:diffform}
    \dot{\mu}_1 =  {\mu}_2 ; \quad
    \dot{\mu}_2 = \Ddot{p}_0(t) ; \quad
    p(t) = {\mu}_1 + v(t),
\end{align}
where ${{\mu} \!\triangleq\! \left [ {\mu}_1~{\mu}_2  \right ]^\top = \left [p_0~ \Dot{p}_0 \right ]^\top \!\in\! \mathbb{R}^2}$ is the state, ${\Ddot{p}_0}$ is the unknown input, and ${ \hat{\mu} \!\triangleq\! \left [ \hat{\mu}_1 ~ \hat{\mu}_2 \right ]^\top \!\in\! \mathbb{R}^2}$ will be the estimation output of a differentiator.

A variety of approaches to real-time differentiation problems are proposed in the literature. For instance, discontinuous signal differentiation algorithms \cite{seeber2021robust}, and high-gain observers \cite{khalil2014high}, \cite{astolfi2018low}. In particular, we consider a class of differentiators that are ISS with respect to perturbations such as noise input:
\begin{definition}[Input-to-state Stable Differentiator] 
\label{def:issbound}
Consider a continuous-time differentiator for system \eqref{eq:diffform} of the form
\begin{equation}
\label{eq:system_dist}
    \dot{\hat{\mu}}  = F \big(\hat{\mu} , p(t) \big),
\end{equation}
where ${F\!:\! \mathbb{R}^2 \!\times\! \mathbb{R}  \!\to\! \mathbb{R}^2}$ is locally Lipschitz in its arguments. 
The differentiator \eqref{eq:system_dist} is an \textit{input-to-state stable (ISS) differentiator} if there exist a ${ \beta \!\in\! \mathcal{KL}}$ and a ${\gamma \!\in\! \mathcal{K}}$ such that for any input ${v \!\in\! L_{\infty}}$ and any initial differentiation error ${\hat{\mu}(0) \!-\! {\mu}(0) }$, the solution of \eqref{eq:system_dist} satisfies for all $t  \!\geq\! 0$:
\begin{equation}
\label{eq:ISS}
     \| \!\underbrace{\hat{\mu}(t) \!-\! {\mu}(t)}_{\triangleq e_{\mu}(t)}  \!\| \!\leq\! \underbrace{\beta(\|{\hat{\mu}}(0) \!-\! {\mu}(0) \|, t) \!+\! \gamma(\|v(t)\|_{\infty})}_{\triangleq \mathcal{M}(t)} ,
\end{equation}
where ${e_{\mu} }$ is the \textit{differentiation error}, and ${\mathcal{M}(t) \geq 0, \forall t \!\geq\! 0}$.
\end{definition}
Definition~\ref{def:issbound} characterizes the performance of differentiator \eqref{eq:system_dist}
in terms of the boundness of the estimation error. For example, differentiation with high-gain observers is ISS \cite{astolfi2018low}. Furthermore, due to the continuity requirement of CBF conditions, high-gain observers are an appropriate differentiator for the rollover prevention problem. A high-gain observer for the class of systems \eqref{eq:diffform} is given by
\begin{equation}
\label{eq:HGO}
    \dot{\hat{\mu}}_1 = {\hat{\mu}}_2 + k_{1} \ell (p(t) - {\hat{\mu}}_1) ; \quad
    \dot{\hat{\mu}}_2 = k_{2} \ell^{2} (p(t) - {\hat{\mu}}_1) ,
\end{equation}
where ${\ell \!\geq\! 0}$ is the high-gain parameter, and ${k_{1}, k_{2}  \!\geq\! 0}$ are the design coefficients.  
The estimation error provided by the observer \eqref{eq:HGO} satisfies the following bound for all $t \!\geq\! 0$:
\begin{equation}
\big\|\hat{\mu}(t) - {\mu}(t)\big\| \leq c_{1}   {\rm e}^{-c_{2}  t} \big\|\hat{\mu}(0) - {\mu}(0) \big\| + c_{3} \|v\|_{\infty}
\end{equation}
for some ${c_1, c_2, c_3  \!\geq\! 0}$; see \cite{khalil2014high}, \cite{astolfi2018low}. 

In our problem setup, we consider safety constraints that rely on multiple time-varying parameters denoted by ${\mathrm{p}_0(t) \!\triangleq\!  [p_{0,1}(t) \ldots p_{0,z}(t) ]^\top}$, ${\mathrm{p}_0 \!:\! \mathbb{R}^+_0 \!\to\! \mathbb{R}^{z}}$, where ${z}$ is the number of parameters needing differentiation, as ${h(x, \mathrm{p}_0(t)) }$, ${h \!:X \!\times\! \mathbb{R}^z \!\to\! \mathbb{R}}$. These parameters are associated with a noisy measurement vector ${\mathrm{p}(t) \!\triangleq\!  [p_1(t) \ldots p_z(t) ]^\top}$, ${\mathrm{p}  \!:\! \mathbb{R}^+_0 \!\to\! \mathbb{R}^{z}}$. We define a new state vector ${x_{\mu} \!\triangleq\!  [ {\mu}_{1,1} \!~ {\mu}_{2,1} \ldots {\mu}_{1,z} \!~ {\mu}_{2,z}  ]^\top \!\in\! \mathbb{R}^{2z}}$, where ${{\mu}_{1,i}\!:\! \mathbb{R}^+_0 \!\to\! \mathbb{R}}$ is a continuously differentiable function, and ${{\mu}_{2,i}\!:\! \mathbb{R}^+_0 \!\to\! \mathbb{R}}$ represents its globally Lipschitz continuous derivative as in \eqref{eq:diffform} for ${i \!=\! 1, \ldots, z}$. And, ${ \hat{x}_{\mu} \!\triangleq\!  [ \hat{\mu}_{1,1} ~ \hat{\mu}_{2,1} \ldots \hat{\mu}_{1,z} ~ \hat{\mu}_{2,z}]^\top \!\in\! \mathbb{R}^{2z}}$ is the estimation output vector of a differentiator. We assume the parameters are differentiated separately using the same ISS differentiator structure. Therefore, we have a multi-input multi-output differentiator dynamics: ${\mathrm{F} \!\triangleq\! \big [F ( \hat{\mu}_{1,i} , \hat{\mu}_{2,i} , p_i(t) ) \big ]}$, ${\mathrm{F} \!:\!\mathbb{R}^{2z} \!\times\! \mathbb{R}^{z}  \!\to\! \mathbb{R}^{2z}}$. 

As the upper bound function ${\mathcal{M}(t)}$ in \eqref{eq:ISS} is valid for a single parameter, but we have multiple differentiated parameters, we need to obtain the maximum of ${\mathcal{M}_i(t)}$, representing ${\mathcal{M}(t)}$ for ${i \!=\! 1,\ldots,z}$, at each time step. To construct a smooth function representing the maximum of ${\mathcal{M}_i}$, we can employ a smooth maximum given by (with  ${\lambda \!\geq\! 0}$):
\begin{equation}
\label{eq:softmax}
\mathrm{M}(t) = \lambda \log \Big(\sum_{i=1}^{z} {\rm e}^{\lambda \mathcal{M}_i(t)} \Big) .
\end{equation}

Now, we define a disturbed augmented system dynamics formed by \eqref{eq:systemdist} and \eqref{eq:system_dist} as
\begin{equation}
\label{eq:sys_aug}
\underbrace{\begin{bmatrix}
\dot x \\ 
\dot{\hat{x}}_{\mu}
\end{bmatrix}}_{\triangleq \Dot{\tilde{x}}}
=
\underbrace{\begin{bmatrix}
f(x) \\ \mathrm{F}(\hat{x}_{\mu} , \mathrm{p}(t))
\end{bmatrix}}_{\triangleq \tilde{f}(\tilde{x}, \mathrm{p}(t))}  +
\underbrace{\begin{bmatrix}
g(x) \\ 
0 
\end{bmatrix}}_{\triangleq \tilde{g}(\tilde{x})} u +
\underbrace{\begin{bmatrix}
d(t) \\ 
0 
\end{bmatrix}}_{\triangleq \tilde{d}(t)}.
\end{equation}
Next, we incorporate the ISS differentiator \eqref{eq:system_dist} into the CBF construction with the augmentation of the existing CBF $h$ by replacing ${\mathrm{p}_0(t)}$ in ${h(x, \mathrm{p}_0(t)) }$ with ${\hat{x}_{\mu}}$. By the Lipschitz continuity of ${h}$, there exists a constant ${L_h \!\geq\! 0}$ that satisfies:
\begin{align}
\begin{split}
    \label{eq:lipsh}
    \!\!\!\!\left | h(x,{\hat{x}}_{\mu}) \!-\! h(x,x_{\mu})   \right | &\!\leq\! L_h \left \| {\hat{x}}_{\mu} 
 \!-\! x_{\mu} \right \| \\
    \implies h(x,x_{\mu}) &\!\geq\! h(x,{\hat{x}}_{\mu}) \!-\! L_h \left \| {\hat{x}}_{\mu}  \!-\! x_{\mu} \right \| \\
    &\!\geq\! h(x,{\hat{x}}_{\mu}) \!-\! L_h \mathrm{M}(t) \!\triangleq\! h_{\mathrm{M}} (t, \tilde{x}),
\end{split}
\end{align}
for any ${(t,\! x,\!x_{\mu},\!{\hat{x}}_{\mu}) \!\in\! \mathbb{R}^{+}_0 \!\times\! X \!\times\! \mathbb{R}^{2z}\!\times\! \mathbb{R}^{2z} }$. 
\begin{remark}
\label{re:lincon}
If the CBF ${h(x, \mathrm{p}_0(t)) }$ is affine in parameter $\mathrm{p}_0$, i.e., ${h(x,\mathrm{p}_0(t))) \!=\! h(x) \!+\! q^\top \mathrm{p}_0(t)),~q \!\in\! \mathbb{R}^{z}}$, then ${L_h \!=\! q}$ is a Lipschitz constant. 
\end{remark}

Similar to the observer-based CBF method proposed in \cite{agrawal2022safe}, we consider $h_{\mathrm{M}}$ and its $0$-superlevel set to enhance robustness against differentiation errors $e_{\mu}$:
\begin{equation}
\label{hecbf}
    \C_{\mathrm{M}}(t) \!\triangleq\! \{ \tilde{x}  \!\in\! X \!\times\! \mathbb{R}^{2z}  : h_{\mathrm{M}}(t,\! \tilde{x}) \!\geq\! 0 \} ,
\end{equation}
which is a time-varying set. Since ${\mathrm{M}(t) \!\geq\! 0,~\forall t \!\geq\! 0}$, ${ \C_{\mathrm{M}}(t)}$ is a subset of the $0$-superlevel set of ${h(x, \mathrm{p}_0(t)) }$, original safe set. We assume that ${\frac{\partial h_{\mathrm{M}}(t,\! \tilde{x})}{\partial \tilde{x}} \neq 0}$ for all ${\tilde{x}(t) \!\in\! \partial {\C_{\mathrm{M}}}(t)}$. Finally, the following definition incorporates the dynamics of the differentiator into a CBF constraint: 
\begin{definition}[Differentiator-Adaptive CBFs]
\label{def:dCBFs}
Let ${\C_{\mathrm{M}}(t) }$ given in \eqref{hecbf} with an ISS differentiator \eqref{eq:system_dist} be the $0$-superlevel set of a continuously differentiable function ${{h}_{\mathrm{M}}\!:\! \mathbb{R}^+_0 \!\times\! X \!\to\! \mathbb{R}}$ with $0$ as a regular value. Then ${{h}_{\mathrm{M}}}$ is a \textit{differentiator-adaptive control barrier function} (DA-CBF) for system \eqref{eq:sys_aug}, without $d$, on ${\C_{\mathrm{M}}(t)}$, if there exists an ${\alpha \!\in\! \mathcal{K}_{\infty, e}}$ such that ${\forall  \tilde{x}(t)  \!\in\! \C_{\mathrm{M}}(t) }$:
\begin{equation*}
   \sup_{u \in U} \!\big [ \!L_{\tilde{f}} h_{\mathrm{M}} (t,\!\tilde{x}) + L_{\tilde{g}} h_{\mathrm{M}} (t,\!\tilde{x})  u 
 \!-\! L_h \Dot{\mathrm{M}}(t)  \!\big ] \! \!\geq\!  -\alpha (h_{\mathrm{M}}(t,\tilde{x})).
\end{equation*}
\end{definition}

Next, we ensure robust safety for the disturbed augmented system \eqref{eq:sys_aug} via the following theorem by leveraging the notions of DA-CBF and tPSSf.
\begin{theorem}
\label{theo:main2}
Let ${h_{\mathrm{M}} \!:\! \mathbb{R}^{+}_0 \!\times\! X \!\times\! \mathbb{R}^{2z} \!\to\! \mathbb{R}}$ be a DA-CBF for \eqref{eq:sys_aug}, without the disturbance $d$, on its $0$-superlevel set ${\C_{\mathrm{M}}(t)}$ with an ${\alpha \!\in\! \mathcal{K}_{\infty, e}}$. 
Any locally Lipschitz continuous controller ${\mathbf{k}\!:\! X \!\times\! \mathbb{R}^{2z} \!\to\! U}$ satisfying ${\forall \tilde{x}(t)  \!\in\! C_{\mathrm{M}}(t) }$:
\begin{equation}
\label{def:DAcbf2}
   \!\!\!L_{\tilde{f}} {h}_{\mathrm{M}} (t,\!\tilde{x}) \!+\! L_{\tilde{g}} {h}_{\mathrm{M}} (t,\!\tilde{x})  \mathbf{{k}}(\tilde{x})  \!-\! L_h \Dot{\mathrm{M}}(t)   \!\!\geq\!  -\alpha ({h}_{\mathrm{M}}(t,\tilde{x})),
\end{equation}
renders the set $\C_{\mathrm{M}}(t)$ tPSSf for \eqref{eq:sys_aug} with respect to the projected disturbance ${\delta(t,\tilde{x}) \!=\! \frac{\partial h_{\mathrm{M}}(t,\tilde{x})}{\partial \tilde{x}} \tilde{d}(t)}$ if for all ${t \!\geq\! 0}$: $${-L_h \Dot{\mathrm{M}}(t) \!+\! \bar{\delta}(t) \!\leq\! -\alpha (-L_h \mathrm{M}(t)).}$$ 
\end{theorem}
\begin{proof}
As ${h_{\mathrm{M}}}$ is a DA-CBF for system \eqref{eq:sys_aug} on ${\C_{\mathrm{M}}(t)}$, Definition~\ref{def:dCBFs} implies that there exists an ${\alpha \!\in\! \mathcal{K}_{\infty, e}}$ such that 
\begin{equation*}
\sup_{u \in U}  \!\!\bigg [ \frac{\partial h_{\mathrm{M}}(t,\!\tilde{x})}{\partial \tilde{x}} \big (\!\tilde{f}(t,\!\tilde{x}) + \tilde{g}(\tilde{x}) u \!\big ) + \frac{\partial h_{\mathrm{M}}(t,\!\tilde{x})}{\partial t} \!\bigg ] \! \!\!\geq\! -\alpha (h_{\mathrm{M}}(t,\!\tilde{x})),  
\end{equation*}
for all ${\tilde{x}(t)  \!\in\! \C_{\mathrm{M}}(t)}$.
Therefore, any Lipschitz continuous controller ${{u} \!=\! \mathbf{{k}}(\tilde{x})}$ satisfying \eqref{def:DAcbf2} renders system \eqref{eq:sys_aug}, without $d$, safe with respect to the set ${\C_{\mathrm{M}}(t)}$ based on Theorem~\ref{teo:cbfdef}. 

Following a similar argument to that in the proof of Theorem~\ref{theo:main} under the condition ${-L_h \Dot{\mathrm{M}}(t) \!+\! \bar{\delta}(t) \!\leq\! -\alpha (-L_h \mathrm{M}(t))}$, we have \eqref{def:DAcbf2} ${ \!\implies\! }$ ${\Dot{h}(t,\!\tilde{x},\!u) \!\geq\! -\alpha (h(\tilde{x})), ~\!\forall ~\!\tilde{x}(t)  \!\in\! \C(t) }$, where
\begin{equation*}
  \Dot{h}(t,\!\tilde{x},\!u) \!=\! L_{\tilde{f}} h (t,\!\tilde{x}) + L_{\tilde{g}} h (t,\!\tilde{x})  u  + \delta(t,\!\tilde{x}), 
\end{equation*}
thus the closed-loop system \eqref{eq:sys_aug} is tPSSf on $\C_{\mathrm{M}}(t)$ with respect to the projected disturbance ${\delta}$. With this robustness, the condition \eqref{def:DAcbf2} leads to ${\tilde{x}(t) \!\in\! \C(t), \forall t \!\geq\! 0}$ if ${\tilde{x}(0) \!\in\! \C(0)}$, which also implies that ${h(x(t),p_0(t)) \!\geq\! 0,\!~\forall t \!\geq\! 0}$. 
\end{proof}
As the DA-CBF condition is affine in the control input $u$, we can define a differentiator-adaptive safety filter. Under Theorem~\eqref{theo:main2}, given a nominal locally Lipschitz continuous controller ${\mathbf{k_d} \!:\! X \!\to\! U}$, ISS differentiator $F$, DA-CBF $h_{\mathrm{M}}$, and ${\alpha \!\in\! \mathcal{K}_{\infty, e}}$ for system \eqref{eq:sys_aug}, the solution of the following QP, {DA-CBF-QP}, ensures robust safety for system \eqref{eq:sys_aug}: 
\begin{align*}
\begin{array}{lll}
{\mathbf{{k}^*}(t, \tilde{x})= \ }
\displaystyle  \argmin_{{u} \in U} \ {\|{u}-\mathbf{k_d}(x)\|^2}  \\ [2mm]
~~\textrm{s.t.} \  L_{\tilde{f}} h_{\mathrm{M}} (t,\!\tilde{x}) \!+\! L_{\tilde{g}} h_{\mathrm{M}} (t,\!\tilde{x})  u \!-\! L_h \Dot{\mathrm{M}}(t)  \!\!\geq\!  -\alpha (h_{\mathrm{M}}(t,\!\tilde{x})\!). 
\end{array}
\end{align*}

Finally, if ${\alpha (h) \!=\! \alpha_c h}$ with ${\alpha_c \!>\! 0}$, i.e., it is a linear extended class-${\mathcal{K}_{\infty, e}}$ function, we have the following corollary:
\begin{corollary}
\label{theo:affine}
Let ${\alpha \!\in\! \mathcal{K}_{\infty, e}}$ in Theorem~\ref{theo:main2} be a linear class-${\mathcal{K}_{\infty, e}}$ function. If there exist an ${\alpha_c}$ such that:
\begin{equation}
\label{eq:cor1}
\alpha_c \geq 1 \quad \text{and} \quad \Dot{\mathrm{M}}(t) \!\leq\! -\alpha_c  \mathrm{M}(t),\!~\forall t \geq 0,
\end{equation}
then, a sufficient condition for \eqref{def:DAcbf2} is given by
\begin{equation}
\label{def:DAcbf3}
   L_{\tilde{f}} h (t,\tilde{x}) + L_{\tilde{g}} h (t,\tilde{x})  u - \alpha_c \Bar{\delta}(t) \geq  -\alpha_c h(t,\tilde{x}),
\end{equation}
which is independent of $L_h$ and $\mathrm{M}$, and ensures that the disturbed augmented system \eqref{eq:sys_aug} is tPSSf on ${\C_{\mathrm{M}}(t)}$.
\end{corollary}
\begin{proof} 
From the time derivative of ${h_{\mathrm{M}}}$ along \eqref{system} and ${\mathrm{F}}$, i.e., the system dynamics \eqref{eq:sys_aug} without $\tilde{d}$, and \eqref{eq:cor1}, \eqref{def:DAcbf3} we have:
\begin{align*}
   L_{\tilde{f}} h(\tilde{x}) \!+\! L_{\tilde{g}} h(\tilde{x})  u  \!\geq\!\! -\alpha_c h(\tilde{x})
     \!+\! \alpha_c L_h {\mathrm{M}}(t) \!+\! L_h \Dot{\mathrm{M}}(t) \\
    \overset{\text{with} \ \Dot{\mathrm{M}} \leq -\alpha_c  \mathrm{M}}{\Longleftarrow} L_{\tilde{f}} h(\tilde{x}) + L_{\tilde{g}} h(\tilde{x})  u - \alpha_c \Bar{\delta}(t) \!\geq\! -\alpha_c h(\tilde{x});
\end{align*}
therefore ${h_{\mathrm{M}}}$ is a DA-CBF with \eqref{eq:cor1}, \eqref{def:DAcbf3}. 
Using the notion of tPSSf, let's analyze the safety of disturbed system \eqref{eq:sys_aug}. From the time derivative of ${h}$ along \eqref{eq:sys_aug}, and \eqref{eq:cor1}, \eqref{def:DAcbf3} we have:
\begin{align*}
   L_{\tilde{f}} h(\tilde{x}) \!+\! L_{\tilde{g}} h(\tilde{x})  u  \!-\! \alpha_c \Bar{\delta}(t) \!\geq\! -\alpha_c h(\tilde{x})
      \\
    \overset{\text{with} \ \alpha_c \geq 1  }{\Longrightarrow} \Dot{h}(t,\!\tilde{x},\!u) \!=\! L_{\tilde{f}} h(\tilde{x}) \!+\! L_{\tilde{g}} h(\tilde{x})  u \!+\!  \delta(t,\!\tilde{x}) \!\geq\! -\alpha_c h(\tilde{x}).
\end{align*}
\end{proof}
We remark that, based on Corollary~\ref{theo:affine}, \eqref{def:DAcbf3} can be utilized to replace the DA-CBF constraint within the DA-CBF-QP.

\section{Rollover Prevention: Theory and Application}
\label{sec:problem}
This section presents a derivation of the safety constraints for the roll motion of a mobile robot via {\em zero moment point}, also referred to as {\em zero-tilting moment point}, (ZMP) criterion, leading to the formulation of a (time-varying) CBF. We leverage the main theoretic result of this paper to demonstrate rollover prevention experimentally. 

\subsection{Rollover CBF Synthesis}
Mobile robots are difficult to model exactly. In practice, it is common to use a simplified model for the design of a mobile robot controller, such as the following model:
\begin{equation}
\label{eq:uni_nom}
    \underbrace{
    \begin{bmatrix}
    \dot{x}^\mathcal{I} \\
    \dot{y}^\mathcal{I} \\ 
    \Dot{\theta} \\
     \dot{\omega} \\
         \dot{v} 
     \end{bmatrix}}_{\dot{x}}
     = 
     \underbrace{\begin{bmatrix}
     v \cos{\theta}  \\
    v  \sin{\theta} \\
    \omega \\
     -\tau_{\omega} {\omega} \\ -\tau_v {v}  
     \end{bmatrix}}_{{f}(x)} +
     \underbrace{
     \begin{bmatrix}
    0 & 0    \\
     0 & 0  \\
      0 & 0 \\
    0 & \tau_{\omega} \\
    \tau_v & 0 
    \end{bmatrix}}_{{g}(x)}
    \underbrace{
    \begin{bmatrix}
        u_{v} \\ u_{\omega}
    \end{bmatrix}}_{u} + d(t) ,
\end{equation}
where ${d\!:\! \mathbb{R}^+_0 \!\to\! \mathbb{R}^5}$ is the disturbance, ${ \big[ {x}^\mathcal{I} ~ {y}^\mathcal{I} \big]^\top \!\in\! \mathbb{R}^2}$ is the vehicle's planar position with respect to the inertial frame ${\mathcal{I}}$, $\theta$ is the vehicle's yaw angle, ${  v }$ is its linear velocity, ${\omega  }$ is its angular velocity (see Fig.~\ref{fig:robot6d}), and ${1/{\tau_{v}}, 1/{\tau_{\omega}} \!>\! 0}$ represent the time constants of the electromechanical actuation system. This model is adopted from \cite{mourikis2007} by assuming that the center of gravity (CG) of the robot intersects with its center of rotation.
\begin{figure}
\centering
\begin{subfigure}{0.49\columnwidth}
\centering\includegraphics[scale=0.45]{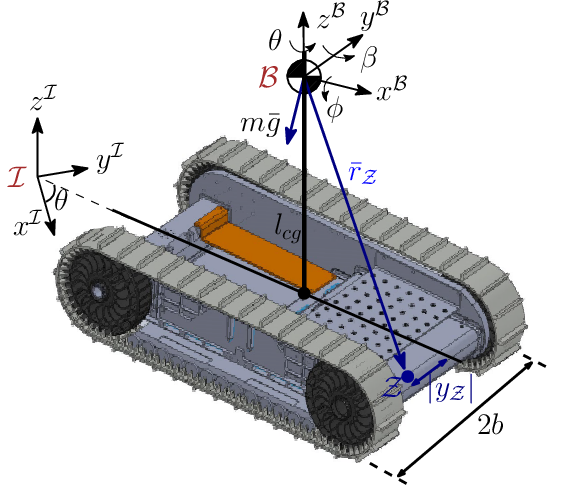}
\end{subfigure} 
\hfill
\begin{subfigure}{0.49\columnwidth}
\centering\includegraphics[scale=0.43]{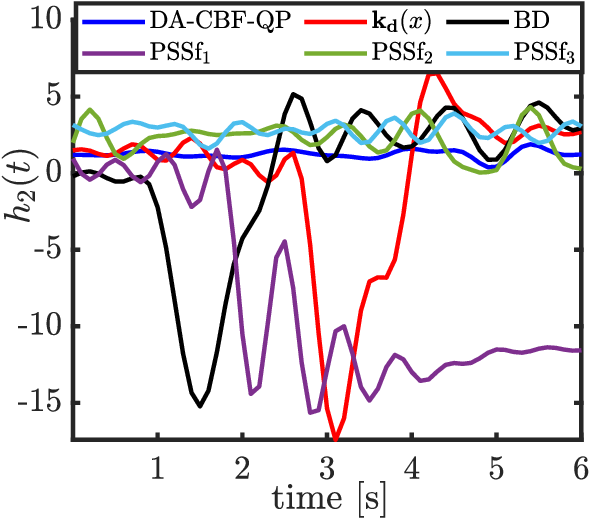}
\end{subfigure} 
\caption{Zero moment point illustration for a mobile robot (Left). The simple backward differentiator (BD) results in noisy values that cause safety violations. The system leaves the safe set for a small and time-invariant projected disturbance bound, ${\Bar{\delta} \!=\! \Bar{\delta}_1}$. Choosing ${\Bar{\delta} \!=\! \Bar{\delta}_2 }$ or ${\Bar{\delta} \!=\! \Bar{\delta}_3 }$, where ${\Bar{\delta}_1 \!<\! \Bar{\delta}_2 \!<\! \Bar{\delta}_3  }$, ensures that ${h_2(t) \geq 0}$, but at the cost of decreased performance due to added conservatives (Right).}
\vskip - 4mm
\label{fig:robot6d}
\end{figure}

The ZMP is the point on the ground where the gravity and inertia forces create only a non-zero moment about the direction of the plane normal, resulting in zero tipping moment \cite{sardain2004}. We compute the mobile robot's ZMP relative to the ground plane and constrain it so that the vehicle does not tip over. The ZMP-based rollover constraint is given by ${\forall t \!\geq\! 0}$:
\begin{equation}
\label{eq:tip}
  | y_\mathcal{Z}(t) | \leq {b} ,
\end{equation}
where ${y_\mathcal{Z}}$ is the lateral component of the ZMP, and $b$ is the half width of the robot. To obtain ${y_\mathcal{Z}}$, we model the orientation of the body-fixed frame relative to the fixed world frame via roll, pitch, and yaw  Euler angles ${\phi,\beta,\theta}$, respectively. The angular rates, angular accelerations, and linear accelerations are given by ${\Bar{\omega} \!=\! \big[
          \Dot{\phi} ~\Dot{\beta} ~ \omega
      \big]^\top ;\ 
      \Bar{\alpha}  \!=\!
      \big[
          \Ddot{\phi} ~ \Ddot{\beta} ~ \Dot{\omega}
      \big]^\top ;\ 
      \Bar{a} \!=\! 
      \big[
\Ddot{x}^\mathcal{B} ~ \Ddot{y}^\mathcal{B} ~ \Ddot{z}^\mathcal{B}
      \big]^\top}$, respectively. 
The robot's rigid body inertia tensor is given by ${I \!=\! \diag(I_{x}, I_{y}, I_{z})}$, and $m$ is total mass. Assuming zero total forces in the ${y}^\mathcal{B}$ and ${z}^\mathcal{B}$ directions, as well as zero moments in the ${x}^\mathcal{B}$ and ${y}^\mathcal{B}$ directions, we have:
\begin{align}
\begin{split}
\label{eq:acc}
\Ddot{y}^\mathcal{B} = -v \!~\omega ; \quad
\Ddot{z}^\mathcal{B} = 0 , \\
I_{x} \Ddot{\phi} + (I_{y} - I_{z})  \Dot{\beta} \!~\omega  = 0.
\end{split}
\end{align}

In Fig.~\ref{fig:robot6d}, $\mathcal{Z}$ is the ZMP point, and ${\Bar{r}_{\mathcal{Z}} \!=\! [x_\mathcal{Z}~y_\mathcal{Z}~\!-\!l_{cg}]^\top }$, where $l_{cg}$ is the distance of the robot's center of mass from the ground. For the sake of simplicity, we also assume that $l_{cg}$ is known. The moment vector about the ZMP is given by
\begin{equation}
\label{eq:zmp}
{M}_{\mathcal{Z}} =  ( \Bar{r}_{\mathcal{Z}} \times m \Bar{a} ) +   ( \Bar{r}_{\mathcal{Z}} \times m  \Bar{g} ) + ( I  \Bar{\alpha} + \Bar{\omega} \times I \Bar{\omega} ), 
\end{equation}
where ${\Bar{g} \!=\! [g_x ~ g_y ~ g_z ]^\top}$ is the gravity vector expressed in the body-fixed frame $\mathcal{B}$. From the definition of the ZMP, the moment at the ZMP must satisfy that
\begin{equation} 
\label{eq:zmpsolve}
   {M}_{\mathcal{Z}} = 
   \begin{bmatrix}
       0&0&M_{\mathcal{Z}_z}
   \end{bmatrix}^\top .
\end{equation}
Then solving \eqref{eq:zmp} with \eqref{eq:zmpsolve} yields:
\begin{equation}
\label{eq:ymp_jwb}
y_{\mathcal{Z}} \!=\! 
\dfrac{ -m \Ddot{y}^\mathcal{B} l_{cg} \!-\! m  g_y l_{cg} \!-\! (I_{x} \Ddot{\phi} \!+\! (I_{y} \!-\! I_{z})  \Dot{\beta} \!~\omega) 
  }{ m \Ddot{z}^\mathcal{B} + m  g_z } ,
\end{equation}
and substituting \eqref{eq:acc} into \eqref{eq:ymp_jwb} yields:
\begin{equation}
\label{eq:yzmp_fin}
y_{\mathcal{Z}} \!=\! 
( v \!~\omega \!~l_{cg} -   g_y \!~l_{cg} )/ g_z  .
\end{equation}
From \eqref{eq:yzmp_fin} and \eqref{eq:tip} we obtain two different time-dependent safety constraints:
\begin{align}
\begin{split}
\label{eq:safety_cbf1}
h_1(t,x) &= v \!~\omega -    {b }/{l_{cg}}\!~{g_z}(t)  - g_y(t) \geq  0   \\ 
 h_2(t,x) &=
-v \!~\omega -    {b }/{l_{cg}}\!~{g_z}(t)  +g_y(t) \geq  0 , 
\end{split}
\end{align}
where ${g_z, g_y}$ are measurable noisy parameters. The function $h_1$ represents safety on the right, while $h_2$ represents the left. Note that ${h_1}$ and ${h_2}$ are affine in ${\left [ {g_y}(t)~g_z(t)\right]^\top}$.
\begin{remark}
\label{re:ICBFs}
In the unicycle model, safety constraints \eqref{eq:safety_cbf1} depend on control inputs $v$ and $\omega$. Inspired by integral CBFs \cite{ames2020integral}, which generalize control input-dependent CBFs, we extend the unicycle model with first-order actuator dynamics as given in \eqref{eq:uni_nom}, where ${u_{v}, u_{\omega}}$ are the new control inputs.
\end{remark}

\subsection{Experimental Validation} 
\label{sec:exper}
We apply our results to an unmanned ground mobile robot, a tracked GVR-Bot from the US Army DEVCOM Ground Vehicle Systems Center. Our Python and C++ algorithms run on a custom compute payload that is based on an NVIDIA Jetson AGX Orin. Vision is provided by three synchronized Intel Realsense D457 depth cameras, and a Vectornav VN-100, an inertial measurement unit (IMU), provides inertial measurements. For the test vehicle and onboard computation details, see Section IV in \cite{janwani2023learning}. We conducted experimental tests on an approximately 27$^{\circ}$ inclined surface, which can cause rollover and slip-induced model uncertainties.

We first designed a nominal controller ${\mathbf{k_d}}$:
\begin{equation*}
    \mathbf{k_d}(x) = \begin{bmatrix}
        K_v  d_g  &
        K_{\omega}  {y_g - y^\mathcal{I}}/{d_g} - K_{\omega} \sin{\theta} 
    \end{bmatrix}^\top ,
\end{equation*}
where ${K_v,\!~K_{\omega} \!\geq\! 0}$ are the controller gains, ${x_g,\!~y_g }$ are the goal position of the robot, and ${d_g \!\triangleq\! \|  x_g \!-\! x^\mathcal{I}, \!~ y_g \!-\!y^\mathcal{I} \| }$.
The inputs were constrained such that ${u_v \!\in\! [-3, 3]}$ m/s, and ${u_\omega \!\in\! [-2, 2]}$ rad/s. The goal position ${x_g,\!~y_g }$ is chosen to yield a safety constraint violation when using the nominal controller. The control loop operated at 50 Hz, and the states were estimated by fusing camera data with inertial measurements. The values of ${\tau_{\omega}, \tau_v}$ in model \eqref{eq:uni_nom} are obtained through a system identification process. The differentiator \eqref{eq:HGO} operates on noisy accelerometer signals ${g_z}$ and ${g_y}$.

We compared the proposed method to the PSSf with time-invariant bounds such that ${0 \!<\! \Bar{\delta}_1 \!<\! \Bar{\delta}_2 \!<\! \Bar{\delta}_3  }$. An increased value of $\Bar{\delta}$ results in a wider gap between the set $\C_{\delta}(t)$ and the forward invariant set $\C(t)$. Consequently, conservative trajectories are produced that remain within $\C(t)$. Conversely, a smaller value of $\Bar{\delta}$ moves $\C_{\delta}(t)$ closer to $\C(t)$, but allowing trajectories to escape from $\C(t)$ in the presence of larger disturbances, as can be observed in \eqref{eq:subset}. Additionally, to show the effectiveness of the proposed differentiator-based method, we obtained the derivative of time-varying parameters with the \textit{backward differentiator} that utilizes the last three data points:
\begin{equation}
\label{eq:BD}
\dot{p}_0 (t_n) =  ({3p_{n} - 4p_{n-1} + p_{n-2}})/({2 T_s}),
\end{equation}
where $T_s$ is the sampling rate, $t_n$ is the sampling time, and ${p_{n}}$, ${p_{n-1}}$, ${p_{n-2}}$ denote the last three measurements, respectively. 

The results of the experiments\footnote{See video at: \url{https://youtu.be/Ekek2ikFU24}} are presented in Fig~\ref{fig:main} and Fig~\ref{fig:robot6d}. These figures illustrate that the proposed method assures the safety of the uncertain system by filtering the unsafe nominal controller through the derived CBFs. Furthermore, the robot reaches ${x_g,\!~y_g }$ as the function $\Bar{\delta}(t)$ was designed to be close enough to $\delta$ along the trajectory. However, the robot trajectory using a PSSf approach with $\Bar{\delta}_1$ ($\text{PSSf}_1$) leaves the safe set due to the violation of the assumption: ${|\delta|_\infty { \leq \bar{\delta}_1}}$. Although PSSf with $\Bar{\delta}_2$ and $\Bar{\delta}_3$ ($\text{PSSf}_2$ and $\text{PSSf}_3$) maintain safety, they yield conservative trajectories, which shows the performance improvement of tPSSf compared to PSSf.

The performance of the CBF-QP controller with differentiator \eqref{eq:BD} is also shown in Fig~\ref{fig:robot6d}. Observe that due to sensor noises and differentiation errors, this controller violates safety. This highlights the importance of robust differentiation along a provable convergence guarantee, as achieved by Theorem~\ref{theo:main2}.

\section{Conclusion} 
\label{sec:conc}
This study developed a rollover prevention method for mobile robots using CBFs and ZMP-based safety measures. A robust safety-critical controller was proposed, incorporating the ISS differentiator dynamics and the notion of PSSf. Experiments conducted on a tracked robot demonstrated the effectiveness of the method in preventing rollover.

\begin{spacing}{0.92}
\bibliographystyle{IEEEtran}
\section*{References}
\vspace{-3 mm}
\bibliography{Bib/Journal_bib}
\end{spacing}

\end{document}